\documentclass[runningheads]{llncs}
\usepackage[T1]{fontenc}   
\usepackage{amsmath}
\usepackage{graphicx}

\newcommand{\junk}[1]{}

\begin{document}

\title{On Solving Problems of
  Substantially Super-linear Complexity in $N^{o(1)}$ Rounds
  in the Model of Massively Parallel Computation}

\author{
Andrzej Lingas
\inst{1}
}
\institute{
  Department of Computer Science, Lund University,
Lund, Sweden.
\texttt{Andrzej.Lingas@cs.lth.se}
}
\pagestyle{plain}

\maketitle


\begin{abstract}
We study the possibility of designing $N^{o(1)}$-round protocols for
  problems of substantially super-linear polynomial-time (sequential)
  complexity in the model of Massively Parallel Computation, where $N$
  is the input size.  We show that if the machines are not equipped
  with relatively large local memory and their number does not exceed
  $N$, then the exponent of the average time complexity of the local
  computation performed by a machine in a round (in terms of local
  memory size) in such protocols must be larger than the exponent of
  the time complexity of the given problem.
\end{abstract}


\section{Introduction}
The distributed computational and communication model of Massively
Parallel Computation (MPC) was introduced by Karloff, Suri, and
Vassilvitskii in  2010 \cite{KSV10}.  It mirrors several key features
of modern large-scale computation systems, such as MapReduce, Hadoop,
Dryad, or Spark.  The model emphasizes  the number of communication phases and
the size of local memory while 
the cost of local computations is less important.
Computation and communication proceeds in synchronized
rounds.  Initially, each machine has access to a distinct part of
the input of preferably equal size.  More specifically and ideally, if the
total input size, measured in machine words, is
$N$ and there are $m\le N$ machines,
then each machine is initially assigned a distinct portion of the
input of size roughly $N/m.$ This implies that the size $s$ of local
memory assigned to each machine is at least $N/m.$ In each round,
each machine can receive messages posted by other machines
in the previous round, perform local
computations (limited only by the local memory size), and post
messages to other machines.
The total size of the messages received or sent by a machine
in a round is bounded by the local memory size $s.$
The main objective of the modern protocols
in the MPC model is to minimize the number of rounds required to
solve the input problem. However, this number depends heavily on
the size of local memory available to each machine.

To illustrate the issue, let us consider
the following approximate protocol for the
NP-hard problem of finding $k$ points (called centers) for a set of
$N$ input points in a metric space such that the maximum distance from
an input point to its closest center is minimized \cite{MK15}.  The
protocol is refereed as a $2$-round $4$-approximation algorithm for the
$k$-center problem in the MPC literature and is even used as a subroutine
in an almost optimal MPC protocol for this problem \cite{HZ23}.  In the
first round each of the $m$ machines runs the classic $2$-approximation
algorithm of Gonzalez \cite{Gon85} on the roughly $N/m$ points assigned to it
and sends the $k$ selected centers to the first machine.  Gonzalez's
algorithm picks an arbitrary input point as the first center and then
repetitively adds the input point that is furthest from the current set
of centers to this set until the set contains $k$ points. In the second round,
the first machine runs the $2$-approximation Gonzalez
's algorithm on
the union of the received $k$-sets and outputs the resulting
$k$ centers for the union.
The authors of
\cite{MK15} show that the protocol yields a 4-approximation to the $k$-center
problem for the input point set and observe
that it can be implemented in two rounds if the local memory available  to
each machine has size $s\ge \max\{N/m,mk\}.$ This is a very generous memory
 requirement. For instance, if $m=\Theta (\sqrt N)$ and $s=\Theta
  (\sqrt N)$ then the delivery of the $m$ $k$-center sets to the first machine
  in the first round would already require $\Omega (k)$ rounds. Moreover, since
  the first machine would have only $O(\sqrt N)$ words of local memory,
  it would be able to store only $O(\frac 1k)$ fraction of the
  $k$-sets,  so the delivery  would need to be repeated each time the
  fist machine adds a new point to the final set of $k$ centers.  This
  would result in $\Omega (k^2)$ rounds
  \footnote{In such scenario, the following straightforward
  $(2k-1)$-round implementation of
  the $2$-approximation Gonzalez's algorithm could be better:
  in the $(2i-1)$-th round the first machine selects the $i$-th center
  and if $i<k$ sends it to each other machine, if $i<k$ in the $2i$-th
  round each machine selects a furthest point to the current center set
  among the points assigned to it and sends the selected point to the first
  machine.  It would require only $O(\max\{N/m,m,k\})$ local space.}.
  This example shows how
  strongly the round complexity of
  basically the same parallel algorithm
  depends on the size of allocated local
  memory (of course other parallel algorithms
  for the $k$-center problem using the same space and
fewer rounds might also exist). The example
  strongly suggests that upper bounds on the number of
  rounds in the MPC model should always be reported
  together with the 
  required local memory size $s$ and the number of machines $m.$
  
So far one has succeeded to design very fast, i.e., $O(1)$-round,
protocols for among other things sorting, searching, some geometric and
geometric-graph problems in the MPC model \cite{BC22,GSZ11}. On the
other hand, several problems including connectivity and distinguishing
between 1 and 2 cycles are believed to require a non-constant number
of rounds in the MPC model. Unfortunately, proving super-constant
lower bounds on the number of rounds in the MPC model seems unfeasible
as it would imply a breakthrough in circuit complexity \cite{RVW18}.

In the literature, one can find several fast protocols in an
alternative round-based distributed computation model of {\em congested
  clique}. Roughly speaking, the main difference from the MPC model is
that there is no limit on the size of local memory; however, the data
that a machine can send or receive is more structured. In a single
round, each machine can receive a message of logarithmic size (with
respect to the input size) from each other machine and likewise can send a
logarithmic-size message to each other machine.  As in the MPC model,
the input is approximately equally partitioned among the machines.
The number of machines is typically about the square root of the input
size. The models of MPC and congested clique are, in general, quite
different. For example, as the number of machines increases, the
amount of data that a machine can exchange in a single round tends to decrease
in the MPC model, whereas it increases in the congested clique model.
The two models are most similar when both the number of machines and
the size of local memory (the latter in the MPC model) are about the
square root of the input size. In this regime, an MPC protocol can be
efficiently simulated on the congested clique by using Lenzen's
$O(1)$-round routing protocol on the congested clique \cite{L}. A
reverse efficient simulation requires a large local memory size,
at least proportional to the number of machines \cite{HP}.

There are known fast protocols, even $O(1)$ round
ones, for several dense graph problems \cite{K_21}, for sorting and routing \cite{L}, and for certain
geometric problems \cite{JLLP25} on the congested
clique. Om the other hand, for such fundamental problems
as  matrix multiplication and all-pairs shortest path, only
$O(N^{\delta})$-round protocols are currently known
in this model \cite{CK19},
where $N$ is the input size
and $\delta$ is a positive constant.
Observe that the known sequential algorithms for these
problems run in substantially super-linear time.  Proving
$\Omega(N^{\delta})$  lower bounds on the round
complexity in the congested clique model, for any $\delta >0,$
is currently unfeasible, as such bounds would also
imply a breakthrough in circuit complexity \cite{RVW18}.
Recently, a simple argument was given in \cite{L25} showing
that protocols for problems with sequential time
complexity $O(N^{\ell})$, where $\ell > 1,$
that run in $N^{o(1)}$ rounds on the congested clique with
$O(\sqrt N))$ nodes must 
perform local computations whose average time complexity
has an exponent strictly
larger than $\ell.$ Since the congested clique model imposes no
restrictions on local computation, this result does not rule out
the existence of fast
protocols for problems with substantially super linear
sequential complexity, it merely indicates that,  if
such protocols exist,
then they must be quite involved. In this paper, we generalize and
adopt the argument from \cite{L25} to the MPC model. The following
example illustrates the main idea behind the adapted argument.

Suppose we have some computational problem
for which the best known sequential algorithm runs in
$N^{\psi}$ time, where $N$ is the input size and $\psi > 1.$
\junk{such that the best known
asymptotic running time of a sequential algorithm for this problem is
$N^{\psi},$ where $N$ stands for the input size and $\psi > 1.$}
Now suppose that we attempt to solve this problem in the MPC model, using
$m=\sqrt N$ machines, each equipped with local memory of size
$s=\Theta(N/m)=\Theta(\sqrt N),$ in $O(1)$ rounds.  Clearly, the $m$ machines
must jointly perform $\Omega (N^{\psi})$ units of work within the
$O(1)$ rounds. Hence on average, a machine has to perform $\Omega
(N^{\psi}/m)=\Omega(N^{\psi -\frac 12})$ units of work. Otherwise, there
would exist a faster sequential algorithm for this
problem. Thus, the average time complexity of the local computation
performed by a machine in a single round, expressed in terms of $s,$ is
$\Omega(s^{2\psi-1})$ while the fastest
known sequential algorithm for this problem runs in
$O(N^{\psi})$ time. Note that
$2\psi -1> \psi $ under the assumption $\psi > 1.$ Consequently,
the asymptotic time complexity of the local computation performed by a
machine in a single round on average is higher than that of the
fastest known sequential algorithm for this problem.  While this does
not yield a contradiction with the MPC model - since the model restricts
only local memory, not local computation - it suggests
that an $O(1)$-round solution to this problem in
the assumed MPC setting (if it exists at all) must be highly non-trivial.  In
particular, it indicates that such a protocol
would involve local computations
that are asymptotically substantially more complex (in terms of the
size of local memory) than just solving some sub-problems of
asymptotic time complexity not exceeding that of the original problem.

Note that the argument in the example above still applies even when the
size of local memory $s$ is $N^{w},$ where $w$ is a constant
$\ge \frac 12,$ while the number of machines remains
$\approx \sqrt N$, provided
  that $\frac {\psi -\frac12}{w}> \psi,$ i.e., $w< \frac
  {\psi -\frac 12} {\psi}.$
  Our results follow from a
  straightforward generalization and formalization of the example
  above covering:
  \begin{itemize}
  \item
  arbitrary problems with substantially
  super-linear polynomial-time complexity $O(N^{\psi})$ (at present),
\item
  $N^{o(1)}$-round protocols in the MPC model with $m\le N$
  machines and local memory size $s=N^{w+o(1)},$ where
  $1-\log_Nm\le w < \frac {\psi - \log_Nm} {\psi},$ and
\item
 even $N^{\delta}$-round protocols in this model for sufficiently small
  $\delta$ depending on the sequential time complexity of the problem.
 \end{itemize}
 
\section{The Formalized Argument}
We formulate our argument in the MPC model with $m\le N$
machines, each equipped with local memory of size $s\ge N/m,$ where
$N$ denotes the input size.

We assume that the {\em work} done by a protocol within $t$
rounds in the MPC model is the total sequential time taken by
the local computations of all $m$ machines during these $t$
rounds, measured in the unit-cost Random Access Machine model with computer
words of logarithmic length \cite{AHU}. Specifically, we assume that
receiving or sending $k$ words in a round requires
$\Theta(k)$ work.

For a problem $P$ solvable in sequential time polynomial  in the input
size $N,$ we define the exponent $opt(P)$ of the (sequential) time
complexity of $P$ as the smallest real number not less than $1$, such
that the problem is known to be solvable in $O(N^{opt(P)+\epsilon})$
sequential time for any positive $\epsilon .$

Consider a protocol $B$ that solves the problem $P$ using $t_B(N,m,s)$
rounds and performing total work $w_B(N,m,s)$ in the MPC model with
$m$ machines, each equipped with local memory of size $s$. We define the
lower bound $ave(B)$ on the exponent of the average local time
complexity of the protocol $B$ (for a single machine in a
single round in terms of $s$)
as the smallest real such that $\frac
{w_B(N,m,s)}{m\times t_B(N,m,s)}=O(s^{ave(B)+\epsilon})$ for any positive
$\epsilon .$ In this definition, we assume the upper
bound of $s$ on the maximal worst-case size of the data received
by a machine before the current round, noting that the
initial input data has size
only $N/m \le s.$

\begin{theorem}
  Consider a problem $P$ solvable in polynomial time and a protocol $B$
  that solves an input instance of 
  $P$ of size $N$ using $t_B(N,m,s)$ rounds
  in the model of MPC with $m\le N$ machines equipped
  with local memory of size $s.$
  If $ave(B)$ is well defined
  then the following inequality holds:
 $$ave(B)\ge \frac {opt(P)-\log_Nm-\log_N t_B(N,m,s)-o(1)
  }{\log_Ns}.$$
  \end{theorem}
   \begin{proof}
  Let $w_B(N,m,s)$ be the total work performed by protocol
  $B$ within its $t_B(N,m,s)$ rounds.
  By the definition of $w_B(N,m,s),$ the problem $P$ can be solved
  in $O(w_B(N,m,s))$ sequential time.
  Indeed, we may simulate the protocol sequentially, round by
  round. In each round, we perform the local computations
  of the machines, one after another,
  using $O(s)$ registers 
  to fetch the messages sent to the current machine
  in the previous rounds and to post new messages, respectively.
  Hence, by the definition of $opt(P)$
  and the fact that $P$ can be solved in $O(w_B(N,m,s))$ time,
  we obtain $w_B(N,m,s)\ge N^{opt(P)-o(1)}.$
 
    By the definition of $ave(B)$,
  we therefore obtain
 $s^{ave(B)} \ge \frac {N^{opt(P)-o(1)}}{m\times t_B(N,m,s)}.$
Consequently, we have
$$ N^{ave(B)\log_Ns }\ge N^{opt(P)-\log_Nm-\log_N t_B(N,m,s)-o(1)}.$$
Taking logarithms yields the theorem.
\qed
\end{proof}

\junk{
  $s^{ave(B)} \ge \frac {N^{opt(P)}}{mt_B(N,m,s)}$
\begin{corollary}
If $t_B(N,m,s)\le
N^{o(1)},$ $s \le N^{\frac 12 +o(1)},$
and $m\le sN^{o(1)}$ then
$ave(B)\ge 2opt(P)-1-o(1).$
  In particular, if $opt(P)>1$  then
 $ave(B)> opt(P).$ 
\end{corollary}}

By straightforward calculations, we obtain the following corollaries.

\begin{corollary}
 If $t_B(N,m,s)\le N^{o(1)}$ and  $s=N^{w+o(1)},$ where
$w$ is a constant satisfying $1-\log_Nm  \le w < \frac {opt(P)- \log_Nm } {opt(P)},$
then $ave(B)> opt(P).$ 
\end{corollary}
\junk{
By substituting $N^{\delta}$ for $t_B(N,m,s)$ in Theorem 1 and
performing straightforward calculations, we also obtain the following
corollary.}
\begin{corollary}
  If  $t_B(N,m,s)\le N^{\delta+o(1)}$, where $\delta$
  is a non-negative constant satisfying 
  $\delta <(1-\log_Ns)opt(P)-\log_Nm,$
  then $ave(B)> opt(P).$
\end{corollary}
\junk{\begin{proof}
  By Theorem 1, it suffices to
  solve the inequality stating that the right-hand side of
  the inequality in Theorem 1, with $\delta+o(1)$ substituted for
  $\log_Nt_B(N,m,s)$, is greater than $opt(P)$. The solution follows
  by straightforward calculations with respect to $\delta.$
 \qed
\end{proof}}
\begin{proof}
  By Theorem 1, it suffices to solve the inequality
  obtained by substituting $\delta+o(1)$ for $\log_Nt_B(N,m,s)$
  in the right-hand side of the inequality in Theorem~1
  and requiring that this expression exceeds $opt(P)$.
  The stated condition on $\delta$
  follows by straightforward calculations.
  \qed
\end{proof}

Our first example concerns
matrix multiplication, say $P=MM,$ in the MPC model
with $m=N^{\frac 12 +o(1)}$ and $s=N^{\frac 12 +o(1)}.$
We have
$opt(MM)=\omega/2$, where $\omega$ is the exponent of the fast matrix
multiplication \cite{VXZ24}. Corollary 1  implies that $ave(B) > opt(MM)$
whenever $t_B(N,m,s)=N^{o(1)}$ and $\omega >2.$ The present upper
bound on $\omega $ is $2.371552$ \cite{VXZ24}
which yields $opt(MM)<
1.186.$ Hence, the threshold value of $\delta$ in Corollary 2  must be
below $(1-\frac 12)1.186-\frac 12\le 0.093$ in the case of matrix
multiplication. Note that the fastest known protocol for
multiplying two $n\times n$ matrices
on the corresponding
congested clique with $n=N^{\frac 12}$ nodes
uses $O(n^{0.157})$ rounds
\cite{CK19}, i.e., $O(N^{0.0785})$ rounds.
Thus, the exponent of the round complexity of the MM protocol from
\cite{CK19} lies below the threshold in Corollary 2 for the assumed
MPC model.  However, a straightforward simulation of this MM protocol in
the MPC model would require a much larger local memory size than
that in the assumed MPC model, roughly proportional to the product of
the number of machines with the number of rounds which is
$\approx N^{0.5785}.$
Furthermore, the exponent of the round
complexity of this  protocol is slightly over the analogous threshold
for congested $n$-clique presented in \cite{L25}.
\junk{guaranteeing the exponent of
of the polynomial bounding the average time of local
computation to be larger than that $\omega/2$ of matrix
multiplication.}
  
Generally, the larger $opt(P)$ is, the larger becomes the gap between
$ave(B)$ and $opt(P)$. This is illustrated in our second example,
which concerns the all-pairs shortest path problem (APSP) in an
edge-weighted graph or digraph on $n$ vertices.

Assume $P=APSP$ and again consider the MPC model with $m=N^{\frac 12 +o(1)}$
and $s=N^{\frac 12 +o(1)}.$
No truly subcubic-time  (in the number of vertices)
sequential algorithm for APSP is
known \cite{VW08}.  Thus, the best known upper bound on $opt(APSP)$ is
$1.5.$ For example, Theorem 1 implies $ave(B) > opt(APSP)+0.4$, when
$t_B(N,m,s)=N^{o(1)}$ and $opt(APSP) =1.5.$ The threshold value of
$\delta$ in Corollary 2  must  therefore be $0.25$ for APSP.  The fastest
known protocol for APSP on an $n$-vertex graph in the congested
$n$-clique model uses $\tilde{O}(n^{1/3})$
rounds \cite{CK19}, i.e., $\tilde{O}(N^{1/6})$
rounds. Thus, the exponent of the round complexity of the APSP protocol
from \cite{CK19} is also below the threshold in Corollary 2
for the assumed MPC model.
However, again,
a straightforward simulation of this APSP protocol in
the MPC model would require a much larger local memory size than
that in the assumed MPC model, roughly proportional to the product of
the number of machines with the number of rounds which is
$\approx N^{2/3}.$
Moreover, the exponent of the round complexity of this APSP protocol
is also slightly above the analogous threshold for congested $n$-clique
presented in \cite{L25}. 

If the aforementioned MM and APSP protocols
on the congested $n$-clique from \cite{CK19} ($n\approx N^{\frac 12}$)
could be implemented using
 almost the same asymptotic number of rounds in the MPC model with
 $m=N^{\frac 12 +o(1)}$ and $s=N^{\frac 12 +o(1)}$,
 then the exponent of the average time complexity
 of the local computations in the implementations would exceed the
 exponent of the (known) time complexity of the MM or APSP problem,
 respectively.
   \section{Final Remarks}

   Of course, the focus of the MPC model is round complexity not
   complexity of local computations.  Similarly, the main focus of our
   results is round complexity.  We show that
   if the round complexity of a protocol
   $B$ and the local memory size $s$ satisfy the assumptions of
   Corollary 1 and/or 2, then the exponent $ave(B)$ of the average
   local computation time of $B$ has to be strictly greater than the
   exponent of the time complexity of the input problem $P.$ This does
   not rule out the existence of such an MPC protocol $B$; rather, it
   suggests that any such protocol would necessarily be highly
   non-trivial.

   Although the main idea of our results coincides with that for
   congested clique from \cite{L25} there are substantial differences,
   e.g., the dependence on the local memory size $s$ and more clear
   concept of local computation complexity. In the congested clique
   model, a machine can store all messages received in prior rounds,
   so in \cite{L25}, one considers all the computations performed by a
   single machine during a protocol run as a local computation whose
   input is the union of the messages received during all rounds. For
   these reasons, the thresholds corresponding to those in Corollaries
   1 and 2 are quite different in \cite{L25}.
   
   It is an interesting question if there are any natural
   reasons for relatively high complexity of local
   computations in the MPC model? For instance,
   higher local computation complexity
   might arise  when certain computations are duplicated
   across multiple machines in order to reduce communication.
  However, this scenario can be handled
   by disregarding duplicating computations in the definition
   of $ave(B)$ without affecting the validity of our argument.
   Another natural reason for higher complexity
   of local computations in the MPC model can be the limit
   $s$ on the local memory size, in particular limiting the
   working local space. However, our results can be easily
   generalized to include a relaxed MPC model without
   any limit on local working space, where $s$ solely limits
   the size of data that a machine can take to its computation
   in the next round.
   \junk{Note also that the limit $s$ on the size of local memory
     could be replaced just by the (implied) limit  $s$ on the size of data
     that a single machine can take to the next round, again without
     affecting the validity of the argument.}

     Current research in the MPC model is heavily
focused on designing very "fast" protocols, preferably of $O(1)$
round complexity. This focus effectively restricts attention to
problems with linear or near‑linear sequential time complexity
(assuming the number of machines does not exceed the input size). This
is unfortunate, because problems of substantially super‑linear
sequential time complexity are not less important and, in fact, need even
more parallel speedup.
One reason for this narrow focus is the absence of at least
logarithmic lower bounds on round complexity for natural
problems. $O(1)$-round protocols are easy to publish because they are
nearly optimal, while polynomial‑round upper bounds for problems of
super‑linear sequential time complexity appear hardly attractive because of
the absence of stronger lower bounds on round complexity.
\junk{
Our simple observations identify thresholds: if an algorithm for a
problem of subs
\item
initially super-linear sequential time complexity uses
fewer rounds than these thresholds, then its local computations must,
on average, have a higher time‑complexity exponent than that in the
sequential complexity of the original problem. This does not rule out
algorithms of round complexity below the thresholds, but it suggests
that any such algorithms would need to be highly non‑trivial.}
Since proving such lower bounds is highly
unlikely, the
thresholds from Corollaries 1  and 2  could serve as
practical reference points for evaluating
polynomial‑round upper bounds for problems of substantially
super‑linear time complexity. In this way, they could help broaden the
current narrow focus of MPC research and stimulate progress on a wider
class of problems.
\bibliographystyle{plain}
  \bibliography{spavoronoi}
  
  \vfill
  \end{document}